\documentclass[12pt,a4paper,final]{iopart}

\usepackage{iopams}  
\usepackage[breaklinks=true,colorlinks=true,linkcolor=blue,urlcolor=blue,citecolor=blue]{hyperref}

\usepackage[utf8]{inputenc}
\usepackage[T1]{fontenc}
\usepackage{mathptmx}

\usepackage[english]{babel}
\usepackage{hyperref}
\usepackage[pdftex]{graphicx}
\usepackage{amsthm}

\newcommand{\RR}{\mathbb R}

\newtheorem*{huygens}{Huygens Theorem}
\newtheorem{theorem}{Theorem}
\newtheorem{proposition}[theorem]{Proposition}
\newtheorem{lemma}[theorem]{Lemma}

\begin{document}

\title[Huygens' envelope principle in Finsler spaces]{Huygens' envelope principle in Finsler spaces\\ and analogue
gravity}

\author{Hengameh R. Dehkordi$^1$}
\address{$^1$Instituto de Matem\'{a}tica e Estat\'{\i}stica,
Universidade de S\~{a}o Paulo, \\
05508-090 S\~{a}o Paulo, SP, Brazil}
\ead{hengamehraeesi@gmail.com, hengameh@ime.usp.br}

\author{Alberto Saa$^2$}%
\address{$^2$Departmento de Matemática Aplicada, 
 Universidade Estadual de Campinas, \\
 13083-859 Campinas, SP, Brazil}
  \ead{asaa@ime.unicamp.br}

\begin{abstract}
We extend to the $n$-dimensional case a recent theorem establishing the validity of
the Huygens' envelope principle for wavefronts in Finsler spaces. Our results have direct applications in   analogue gravity models, for which the 
Fermat's principle of least time  naturally gives origin to an underlying Finslerian geometry. 
For the sake of  illustration, we consider two explicit examples motivated by recent
experimental results: surface waves in flumes and vortices. For both examples, we have distinctive 
directional spacetime structures, namely horizons and ergospheres, respectively. We show that 
both structures are associated with certain directional divergences in the
underlying  Finslerian (Randers) geometry. Our results show that Finsler geometry may
provide a fresh view on the causal structure of spacetime, not only in analogue models but also for General Relativity. 
\end{abstract}


\submitto{\CQG}

\section{Introduction}

Wave propagation in non-homogeneous and anisotropic  media has attracted a lot of attention recently 
in the context of analogue gravity. (For a comprehensive review on the subject, see \cite{Rev}.)  Many interesting results  have been obtained, for instance, by
observing surface waves in some specific  fluid flows, especially those ones corresponding to
analogue black holes, {\em i.e.}, flows exhibiting an effective horizon  for wave propagation \cite{A1,A2,A3,A4}. 
Fluid configurations involving vortices,  which could  in some situations  exhibit     effective ergospheres,  have also been investigated \cite{A5,A6,A7,A8}. The key idea, which
can be traced back to the seminal work \cite{Unruh} of W. Unruh in the early eighties, is the observation
that   generic perturbations $\phi$ in a perfect  fluid of density $\rho$ and with a velocity   field $V = (v^1,v^2,v^3)$ are effectively governed by the Klein-Gordon equation
\begin{equation}
\label{KG}
\frac{1}{\sqrt{-g}}\partial_a \sqrt{-g}g^{ab}\partial_b \phi = 0,
\end{equation}
where $a,b=0,1,2,3$, with the effective metric $g_{ab}$   given by
\begin{equation}
\label{effective}
ds^2 = g_{ab}dx^a dx^b = \frac{\rho}{c}\left( -c^2dt^2 + \delta_{ij}\left(dx^i + v^i dt \right) \left(dx^j + v^j dt \right)\right),
\end{equation}
where
$i,j = 1,2,3$, and $\delta_{ij}$ stands for the usual Kronecker symbol. In general configurations, both the perturbation propagation velocity $c$ and the flow velocity field $V$ can indeed depend upon  space and time, but we are only  concerned here with the stationary situations, {\em i.e.},   the cases   $c=c(x)$ and $V=V(x)$. The spacetime  hypersurfaces corresponding to $c^2 = V^2$, where
$V^2 = v_iv^i = \delta_{ij}v^i v^j$, mimic, from the kinematic point of view, many distinctive properties of the Killing horizons in General Relativity (GR) \cite{Rev}, and this fact is  precisely  the starting point of many interesting analogue gravity studies. For a review on the causal structure of
analogue gravity models, see \cite{Horiz}. The region where $c^2 > V^2$ is the analogue of the exterior region
 of a black hole in GR, where the observers are expected to live. The
null geodesics of (\ref{effective}) correspond to the characteristic curves of the hyperbolic partial 
differential equation (\ref{KG}) and, hence, they play a central role in   the time evolution of their solutions.

For the null geodesics of (\ref{effective}). {\em i.e.}   the curves such that $ds^2=0$,  we
have in the exterior region  
\begin{equation}
\label{Randers}
dt = F(x,dx^i) = \sqrt{a_{ij}(x)dx^i dx^j} + b_i(x)dx^i ,
\end{equation}
where
\begin{equation}
a_{ij} = \frac{(c^2 - V^2)\delta_{ij}+v_iv_j}{(c^2 - V^2)^2} \quad {\rm and} \quad
b_i = \frac{v_i}{c^2 - V^2}.
\end{equation}
Notice that  the formulation 
  (\ref{Randers}) of the null geodesics is, in fact, equivalent to 
  the Fermat's principle of least time, in the sense that the $x^i(s)$ spatial curves
minimizing the time interval $\int dt$ correspond to null geodesics of the original four-dimensional spacetime metric (\ref{effective}).
 On the other hand, the metric defined by (\ref{Randers}) is an explicit example of a   well-known structure in Finsler geometry called Randers metric. Its  striking difference when compared with the usual
 Riemannian metric is that, for $b_i\ne 0$, $dt^2$ is not a quadratic form in $dx^i$, implying
 many distinctive properties for the underlying geometry as, for instance, that 
$ F(x,dx^i) \ne F(x,-dx^i)$, leading to widespread assertion  that, in general, distances do depend
on   directions in Finsler geometries. 
  For some 
 interesting historical notes on this matter, see \cite{Chern}. General Relativity has some
 emblematic  examples of directional spacetime structure as, for instance, event horizons, {\em i.e.} (null)-hypersurfaces which can be crossed only in one direction. It is hardly a surprise that Finsler
 geometry turns out to be relevant for these issues, but the Finslerian description of such spacetime structures from the physical point of view is still a rather incipient program. The present paper is a small step towards such a wider goal.

Finsler geometry is a centenary topic in Mathematics \cite{Finsler}, with a quite large 
accumulated literature. 
The recente review \cite{Shen} covers all   pertinent concepts for our purposes here.  Since
its early days, Finsler geometry has been applied in several contexts, ranging from the already classical control problem known as the Zermelo's navigation problem (see \cite{bao2004zermelo} for a recent approach and  further references) to recent applications in the Physics of graphene \cite{cvetivc2012graphene}. 
 The propagation of wavefronts in different situations and the description of some causal structures of the underlying 
 spacetime from a Finslerian point of view, which indeed are the main topics of the present paper, have been already considered in  \cite{anderson1982modelling,gibbons2009stationary,gibbons2011geometry,
 javaloyes2012conformally,M1,markvorsen2016finsler,M2}. In particular, Markvorsen
proved  in \cite{markvorsen2016finsler} that  the Huygens' envelope principle for wavefronts holds for generic two-dimensional Finsler geometries. Some clarifications are necessary here. We call as the   Huygens' envelope principle for wavefronts the following statement, which is presented as the Huygen's theorem in,
for instance, \cite{arnol2013mathematical}.  
\begin{huygens} 
	Let $\phi_p(t)$ be a wavefront,  which started at the point $p$, after time $t$. For every point $q$ of this wavefront, consider the wavefront after time $s$, i.e. $\phi_q(s)$. Then, the wavefront of the point $p$ after time $s+t$, $\phi_p(s+t)$, will be the envelope of the wavefronts $\phi_q(s)$, for every $q\in\phi_p(t)$.
\end{huygens}
\noindent Such property is rather generic and it is  valid, for instance, for any kind of linear waves in flat spacetime.
Indeed, it was proved in \cite{arnol2013mathematical} for all waves obeying the Fermat's principle of least time
in Euclidean space. It is also verified,  in particular, for the solutions of the Klein-Gordon equation (\ref{KG}) in flat spacetimes of any dimension.
The very fundamental concept of   light cone in Relativity is heuristically constructed from this kind of wavefront
propagation, for which the Huygens' envelope principle is expected to hold on physical grounds. 
 However, such principle should not be confused with the more stringent and restrictive   Huygens' principle which implies that,   besides of the property of the envelope of the wavefronts,   the wave propagation occurs sharply only along the characteristic curves, implying, in particular, the absence of wave tails. Such more restrictive Huygens' principle   is verified, for instance,  for the solutions of the Klein-Gordon equation (\ref{KG}) in flat spacetimes only for odd spatial dimensions. For further details on this issue, see \cite{huygensBook}. 
Provided that the wavefronts  satisfy  the Huygens' envelope principle, one can determine   the time
evolution of the wavefronts for $t>t_0$ once we know the wavefront at $t=t_0$. In this sense,
the behavior of the propagation is completely predictable solely with the information of the wavefront 
at a given time. It is important to stress that one cannot take for granted Huygens' theorem in
a Finslerian framework due to inherent  intricacies of the   geodesic flow. 

In this paper, we will explore some recent mathematical results \cite{raeisidehkordi2018finsler,alexandrino2018finsler} to present a novel proof extending, for the $n$-dimensional case, the Markvorsen result on the   Huygens' envelope principle in generic Finsler spaces. Moreover, we show, by means of 
some explicit examples in analogue gravity, that the Finslerian formulation of the wavefront propagation  
in terms 
of a Randers metric can provide useful insights on the causal structure of the underlying spacetime. 
In particular, we show that the distinctive directional properties of analogue horizons and ergospheres 
have a very natural description in terms of Finsler geometry. In principle, the same Finslerian description would be also available for General Relativity.

We will start, in the next section, with a brief  review on the main mathematical definitions and properties of Finsler
spaces and Randers metrics. Section 3 is devoted to the new proof of the Huygens' theorem and to the discussion of 
  some
generic properties of the geodesic flow and wavefront propagation in $n$-di\-men\-sional Finsler spaces. The two explicit examples motivated by the common hydrodynamic  
analogue models, the cases of surface waves in flumes and
vortices, are presented in Section 4.
The last section is left for some concluding remarks
on the relation between the causal structure of 
spacetimes  and the Finslerian structure of the underlying geometry associated with the
Fermat's principle of least time.
	
\section{Geometrical Preliminaries  }
\label{prelim}

For the sake of completeness, we will present here a brief review on Finsler geometry and the Randers metric,
with emphasis on the notion of transnormality\cite{raeisidehkordi2018finsler,alexandrino2018finsler},
which  will be central in our proof of Huygens' theorem. 
For further definitions and references, see \cite{Shen}. 

Let $V$ be a real finite-dimensional vector space. A non-negative function $F:V \to [0 , \infty)$ is called a Minkowski norm if  the following properties hold:
	\begin{enumerate}
		\item  $F$ is smooth on $V\backslash \{0\}$,
		\item $F$ is positive homogeneous of degree 1, that is $F(\lambda y) = \lambda F(y)$ for every $\lambda > 0$,
		\item for each $y\in V\backslash \{0\}$, the fundamental tensor $g_y$, which is the symmetric bilinear form defined as  
		\begin{equation}\label{fun.form}
		g_y(u,v) = \frac{1}{2}\left( \frac{\partial^2}{\partial t \partial s}F^2(y+tu +sv) \right)_{s=t=0}, 
		\end{equation}	
		is positive definite on $V$.
	\end{enumerate}
The pair $(V,F)$ is usually called a Minkowski space in Finsler geometry literature, and this, in principle, might cause some confusion with the distinct notion of Minkowski spacetime. Here, we will adopt the Finsler geometry standard denomination and no confusion should arise since we do not mix the two different spaces. 
Given a Minkowski space, the indicatrix of $F$ is
the unitary geometric sphere in $(V,F)$, {\em i.e.}, the
 subset 
\begin{eqnarray}
\mathcal{I}=\{v\in V\ |\ F(v)=1\}.
\end{eqnarray}
The indicatrix $\mathcal I$ defines a hypersurface (co-dimension 1) in $(V,F)$ 
consisting of the collection of the endpoints of   unit tangent vectors. In contrast with the Euclidean
case, where the $\mathcal I$ is always a sphere, it can be a rather generic surface in 
a Minkowski space. 
We are now  ready to introduce the notion of a Finslerian structure on a manifold. Let $M$ be an $n$-dimensional
differentiable manifold and $TM$ its tangent bundle. A Finsler structure on $M$ is a 
 function $F:TM\to [0,\infty)$ with the following properties:
	\begin{enumerate}
		\item  $F$ is smooth on $TM\backslash \{0\}$,
		\item  For each $x\in M, \ F_x =F|_{T_xM}$ is a Minkowski norm on $T_xM.$ 
	\end{enumerate}
The pair $(M,F)$ is called a Finsler space.
Suppose now that $M$ is a Riemannian manifold endowed with metric $\alpha :{TM}\times{TM} \to [0,\infty)$
and  a 1-form $\beta:{TM}\to\mathbb{R}$
such that $\alpha(y_\beta,y_\beta)<1$,
with $y_\beta$ standing for the vector dual of $\beta$. In this case,  $F=\alpha+\beta$ is a particular Finsler structure 
called {Randers metric} on $M$, and in this case the pair $(M,F)$ is   called a  {Randers space}.

It is interesting to notice that every Randers metric is associated with a Zermelo's navigation
problem \cite{bao2004zermelo}. Such a problem is defined on 
  a Riemannian space $(M,h)$ with  a smooth vector field (wind)  $W$  such that   $h(W,W) < 1$.  The associated Randers metric corresponding to the solution of a Zermelo's   navigation  problem is given by 
\begin{equation}
\label{Randers1}
  F(y)= \alpha(y) + \beta(y)= \frac{\sqrt{h^2(W,y)+\lambda h(y,y)}} {\lambda}- \frac{h(W,y)}{\lambda}   
 \end{equation} 
where $\lambda =1-h(W,W) $. Comparing with (\ref{Randers}), one can easily establish a conversion between the so-called Zermelo data $(M,h,W)$ of a Randers space and the analogue gravity quantities $c$, $V$, and $\delta_{ij}$.

Given a Finsler space $(M,F)$, the gradient $\nabla f_p$ of a 
 smooth function $f:M\to\RR$  at point $p\in M$
  is defined as
  \begin{equation}
  \label{nabla}
df_p(v)=g_{\nabla f_p}(\nabla f_p, v),
  \end{equation}  
where $v \in T_pM$ and 
\begin{equation}
g_y(y,v)=\frac{1}{2}\left(\frac{\partial }{\partial s} F^2(y+sv) \right)_{s=0},
\end{equation}
 which is the fundamental tensor of $F$ at $y\in T_pM$ (see \cite{javaloyes2011definition} for more details). It is important to stress that, in Randers spaces, where a Riemannian structure is
also always available, the gradient (\ref{nabla}) differs  from the usual
 Riemannian gradient $\tilde\nabla f_p$ at $p\in M$, unless the vector field $W$ vanishes. The following Lemma, which proof can be found in \cite{alexandrino2018finsler}, connects the two gradients in a very useful way, since the
 direct calculation of $ \nabla f_p$ 
 is sometimes rather tricky. 
 \begin{lemma}
 \label{nablas} 
 Let $f:U\subset M \to \mathbb{R}$ be a smooth function without critical points,  
   $(M,F)$   a Randers space with Zermelo data $(M,h,W)$, and
$\nabla f_p$ and $\tilde\nabla f_p$, respectively, the gradients with respect to $F$ and to $h$
at $p\in M$. Then
\begin{enumerate}
\item $\displaystyle \frac{||\tilde\nabla f_p||}{F(\nabla f_p)}\left(\nabla f_p -
F(\nabla f_p)W \right) = \tilde\nabla f_p$,
\item  $\displaystyle F(\nabla f_p) = ||\tilde\nabla f_p|| + df(W) $,
\end{enumerate} 
where $||y||^2 = h(y,y).$
 \end{lemma}

 If $L$ is a submanifold of a 
 Finsler space $(M,F)$, a non-zero vector $y\in T_pM$ will be orthogonal to $L$ at $p$ if $g_y(y,v)=0$  for every $v\in T_pL.$ Notice that, for the case of a Randers space 
	with   Zermelo's data $(M,h,W)$, for every non-zero vectors $u$ and $y$ in $T_pM$, we will have
	$g_y(y,u)=0$ if and only if  (see Corollary 2.2.7 in \cite{raeisidehkordi2018finsler})
\begin{equation}
\label{orto}
 h\left(u,\frac y{F(y)}-W\right)=0.
 \end{equation}
The following Lemma, which proof follows straightforwardly  from the previous definitions (see also \cite{Shen}), will be useful
in the next section. 	

\begin{lemma}
	Let $(M,F)$ be a Finsler space, $\mathcal{U}$ an open subset of  $M$, and $f$ a smooth function on $\mathcal{U}$ with $df\neq 0$. Then, $ n=\left. \frac{\nabla f}{F(\nabla f)}\right|_{f^{-1}(c)}$ is orthogonal to $f^{-1}(c)$ with respect to $g_n$. \label{ortho}
\end{lemma}

We can now introduce the notion of transnormality in Finsler spaces, a   concept that
has begun to attract some considerable attention  in geometry,
see \cite{transnormal}, for instance.
Let $f:M\to \RR $ be a smooth function. If there exists a continuous function $\mathfrak{b}:f(M) \longrightarrow \RR$ such that 
\begin{equation}
F^2(\nabla f)=\mathfrak{b}\circ f, 
\end{equation} 
with $\nabla f$ given by (\ref{nabla}),
 then $f$ is called a  {Finsler transnormal} (shortly $F$-transnormal) function. When transnormal functions are available, 
some properties of the geodesic flow in a Finsler space can be easily determined. Geodesics
in Finsler geometry are defined in the same way of Riemannian spaces. First, notice that 
the  {length} of a piecewise smooth curve $\gamma: [a,b] \longrightarrow M$ with respect to $F$ is defined as
\begin{equation}
\label{length}
\mathit{L}(\gamma)=\int_{a}^{b}F(\gamma(t),\gamma'(t))dt.
\end{equation}
Analogously to the Riemannian case, the  {distance} from a point $p\in M$ to another point $q\in M$ in the Finsler space $(M,F)$ is given by
\begin{equation}
\label{d_F}
d_F(p,q)=\inf_\gamma\int_{a}^{b}F(\gamma(t),\gamma'(t))dt,
\end{equation}
where the infimum is meant     to be taken over all piecewise smooth curves $\gamma:[a,b] \longrightarrow M$ joining $p$ to $q$. For a Finsler space $(M,F)$, the  {geodesics} of $F$ are the length (\ref{length}) minimizing curves. Notice that, when we are dealing with  Randers spaces   derived from the null geodesic of a Lorentzian manifold, as it was discussed in Section 1, the geodesics of  $(M,F)$ correspond to a realization of Fermat's principle of least time for the original null geodesics.  For a more general mathematical discussion
on the Fermat's principle in Finsler geometry, see \cite{M2}.
It is worth mentioning that, for some special vectors $W$, there is a useful relation between   geodesics in a  Randers space  
with Zermelo data 
$(M,h,W)$ and the usual geodesics in the Riemannian space $(M,h)$. Such relation is expressed 
by the following Lemma, which follows directly as a Corollary of Theorem 2 in \cite{robles2007geodesics}.
	\begin{lemma}\label{geo}
	 Let $(M,h)$  be a Riemannian manifold endowed with a Killing vector field $W$. 
		Given a unitary geodesic $\gamma_h:(-\epsilon,\epsilon)\to M$ of $(M,h)$, the curve ${\gamma}_F(t)=\varphi_W(t,\gamma_h(t))$, where $\varphi_W:(-\epsilon,\epsilon)\times U \to M$ is the flow of $W$, will be a 
		 F-unitary geodesic  of the Randers space $(M,F)$ with Zermelo data $(M,h,W)$.
\end{lemma} 

The  {distance} from a given compact subset $ A $ of a manifold $ M $ to any point $ p\in M $ is defined as $\rho:M\to\RR$ with $\rho(p)=d_F(A,p)$. If for every $p,q \in M$ there exists a shortest unit speed curve from $p$ to $q$, then $F(\nabla \rho)=1$, indicating that $\rho$ is $F$-transnormal with $\mathfrak{b}= 1$ \cite{Shen}. The next results, which proofs can be found in \cite{alexandrino2018finsler}, will be useful to 
characterize the relation between the propagation of wavefronts  and   transnormal function in Finsler spaces.

\begin{proposition}[]\label{parallel}
	Let $f:M\to \RR$ be a $F$-transnormal function with $f(M)=[a,b]$. If  $c<d \in f(M)$, then for every $q\in f^{-1}(d)$,  $$d_F(f^{-1}(c),q)=d_F(f^{-1}(c),f^{-1}(d))=\int_{c}^{d}\frac{ds}{\sqrt{\mathfrak{b}(s)}}= L(\alpha),$$ where $\alpha$ is a  reparametrization of (an extension of) the integral curve of $\nabla f$.
\end{proposition}

\noindent Notice that, from this proposition, we have   $f^{-1}(c)\subseteq \rho^{-1}(r)$ where $\rho(p)=d_F(f^{-1}(a),p)=r$. We say that two submanifolds $C$ and $D$ of a Finsler space are 
equidistant  if, for every $p\in C$ and $q\in D$, $d_F(p,D)=d_F(C,D)$ and $d_F(D,C)=d_F(q,C)$ (or, equivalently, $d_F(C,D)=d_F(C,q)$ and $d_F(D,C)=d_F(D,p)$).

 \begin{theorem}\label{improv}
	Let $M$ be a compact manifold and $f:M\to \mathbb{R}$ be a $F$-transnormal and analytic function such that $f(M)=[a,b]$. 
	Suppose that the level sets of $f$ are connected and $a$ and $b$ are the only critical values of $f$ in $[a,b]$. 
	Then,
	 for every $c,d \in [a,b]$, $f^{-1}(c)$ is equidistant to $f^{-1}(d).$ 
\end{theorem}
\noindent Finally, the {\em cut loci} of the point $p$ associated to the distance function $\rho$ is defined analogously 
to the Riemannian case: it consists in the set of all points $q\in M$ with
two or more different  length (\ref{length}) minimizing 
curves  $\gamma:[a,b] \to M$ joining $p$ to $q$.

\section{The Huygens' envelope principle in Finsler spaces}
\label{propagationsec}

Throughout this section, it is assumed that, on some part of a Finsler space $(M,F)$, a wavefront is spreading and sweeping the domain $U\subset M$    in the interval of time from $t=0$ to $t=r$. 
It is also assumed that $ U $ is  a smooth manifold. Given a wavefront $\phi_p(t)$, we call the wave
ray at $q\in \phi_p(t)$ the shortest time path connecting $p$ to $q$. Again, due to the intricacies 
of Finslerian metrics, one cannot take for granted many properties of wave rays in Euclidean spaces as, for instance, the fact that they are orthogonal to the wavefronts. Let us start by considering 
 the Huygens' theorem for more general situations. 
The following theorem generalizes Markvorsen's result \cite{markvorsen2016finsler} for any Finsler space.

\begin{theorem}\label{propag.fire} 
	Let $\rho:M\to \RR$ with $\rho(p)=d_F(A,p)$, where $A$ is a compact subset of $M$ and $\rho(U)=[s,r]$, where $0<s<r$. Suppose that $\rho^{-1}(s)$ is the wavefront at time $t=0$ and that there are no cut loci in $\rho^{-1}([s,r])$. Then,
		for each $t\in [s,r]$, $\rho^{-1}(t)$ is the wavefront at time $t-s$ and 
	the Huygens' envelope principle is satisfied by all the wavefronts 
	$\displaystyle \{\rho^{-1}(t)\}_{t\in [s,r]}.$
	 Furthermore, the wave rays are geodesics of $F$ and they are also   orthogonal to each wavefront $\rho^{-1}(t)$ at time $t-s$.
\end{theorem}
\begin{proof}
	Since $\rho$ is a transnormal function with $\mathfrak{b}=1$, from Proposition \ref{parallel} 
we have that, for every $t>s$ and $q\in \rho^{-1}(t)$,   
	 \begin{equation}\label{firefont}
	d_F\left(\rho^{-1}(s),q\right)=t-s,
	\end{equation} 
meaning that the wavefront  reaches  $\rho^{-1}(t)$ after time $t-s$. The relation $d_F(\rho^{-1}(s),q)= d_F(\rho^{-1}(s),\rho^{-1}(t))$ implies that no part of the wavefront meets $\rho^{-1}(t)$ before time $t-s$,  and thus $\rho^{-1}(t)$ is indeed the wavefront at this time. 
	
	Now, in order to
	 verify the Huygens' envelope principle, let us assume that $e(\delta)$ is the envelope of radius $\delta$ of the wavefront $\rho^{-1}(t_0)$ for some time $t_0\geq s$. It implies that for every $p\in e(\delta)$, 
	\begin{equation}\label{env}
d_F(\rho^{-1}(t_0),p)=d_F(\rho^{-1}(t_0),e(\delta))=\delta,
	\end{equation}
which follows from contradiction, since if  	
 there would exist some point $p_0\in e(\delta)$ and $q_0\in \rho^{-1}(t_0)$ such that $d_F(\rho^{-1}(t_0),p_0)=d_F(q_0,p_0)=r<\delta$, then the  wavefront centered at $q_0$ and radius $\delta$
 would  intersect the envelope, which is a contradiction and consequently relation \ref{env} is indeed valid.
  So, as $p\in e(\delta)$, there exists a path from a unique point $q\in \rho^{-1}(t_0)$ to the point $p$ along which the wavefront time of travel   is precisely $\delta$. Since $\rho^{-1}(t_0)$ is the wavefront,
 this wave ray  has emanated from some point in $\rho^{-1}(s)$ and reached   point $q$ at time $t_0-s$. 
	Therefore, we have
\begin{equation}
	 d_F(\rho^{-1}(s),p)\leq t_0-s+\delta. 
\end{equation}	
Notice that, if $d_F(\rho^{-1}(s),p)<t_0-s+\delta$, there would exist a path from $\rho^{-1}(s)$ to $p$ through which wave ray travels in a time shorter than $t_0-s+\delta$. As 
\begin{equation}
 d_F(\rho^{-1}(s),\rho^{-1}(t_0))=t_0-s, 
\end{equation}
 this ray meets $\rho^{-1}(t_0)$ at exactly time $t_0-s$. As a result, the inequality would hold only when this ray travels from $\rho^{-1}(t_0)$ to $p$ at a time less than $\delta$ which is a contradiction by Eq. (\ref{env}). Finally, we have 
\begin{equation}
d_F(\rho^{-1}(s),p)=t_0-s+\delta =t-s
\end{equation} 
  which means $p$ belong to the wavefront $\rho^{-1}(t)$, and hence $e(\delta)\subset\rho^{-1}(t)$. 
	Now, we can establish that $\rho^{-1}(t)\subset e(\delta)$. Assume that $p\in\rho^{-1}(t)$. 
	Since $\rho^{-1}(t_0)$ is the wavefront, each wave ray from $\rho^{-1}(s)$ reaches   $\rho^{-1}(t_0)$ and $\rho^{-1}(t)$,   at times $t_0-s$ and $t-s$, respectively. 
	Using Proposition \ref{parallel}, one has
\begin{equation}	
	d_F(\rho^{-1}(t_0),p)=t-t_0=\delta,  
\end{equation}	
and consequently $p\in e(\delta)$. 

	To accomplish the proof, observe that each wave ray emanates from a point in $\rho^{-1}(s)$ and reaches   $\rho^{-1}(t)$ in the shortest time, implying that its traveled   path is a geodesic of the Finsler  space. Furthermore, assuming that $\alpha$ is the unit speed geodesic such that
\begin{equation}
 d_F(\rho^{-1}(s),\rho^{-1}(t))=d_F(\rho^{-1}(s),p)=L(\alpha|_{[0,t]})=t,
\end{equation}	
we have, according to Proposition \ref{parallel}, that $\alpha|_{[0,t]}$ is an extension of the integral curve of $\nabla\rho$. Hence, $\alpha|_{(0,t)}$ is the integral curve of $\nabla\rho$, and by Lemma \ref{ortho} it is orthogonal to each $\rho^{-1}(t)$. 
\end{proof}
\noindent Notice that the cut loci in $\rho^{-1}([s,r])$ are associated with singularities in the wavefronts, an extremely interesting topic \cite{sing}, but which is out of the scope of the present paper. 

If a transnormal function $f$ is available, one can determine the wavefronts without dealing with 
the Randers metric and/or the distance function. The following proposition, which proof follows in the same
way of Theorem \ref{propag.fire}, summarize this point. 

\begin{proposition}\label{propag.water}
	Suppose that   $f:M\to\RR$ is a $F$-transnormal function with $F^2(\nabla f)=\mathfrak{b}(f)$ and  $f(M)=[a,b]$. Assuming that  $f^{-1}(a)$ is a wavefront at time $t=0$, we have
	\begin{itemize}
		\item [$a)$] for every $c\in[a,b]$, $f^{-1}(c)$ is the wavefront at time 
\begin{equation}
 r_{a,c}=\int_{a}^{c}\frac{ds}{\sqrt{\mathfrak{b}(s)}} ,
\end{equation}		
		\item [$b)$] $\{f^{-1}(c)\}_{c\in [a,b]}$ satisfies Huygens' envelope principle,
		\item [$c)$] the wave rays are geodesics of $F$ joining  $f^{-1}(a)$ to $f^{-1}(b)$, and
		they are  also orthogonal to each wavefront.
	\end{itemize}
\end{proposition}

\section{Analogue Gravity Examples}

In this section, we will present two explicit examples,
in the context 
of  analogue gravity, 
 of wavefront propagation  
determined from the Huygens' envelope principle in Randers spaces,
whose validity for any space dimension was established by our  mathematical results.
  The examples, motivated by very recent
experimental results, are namely the cases of surface waves in flumes and vortices. Of course, we 
are assuming that for such realistic cases the  surface waves indeed obey a Klein-Gordon equation (\ref{KG}), 
for which the Huygens’ envelope principle is expected to hold on physical grounds.  
Nevertheless, in realistic experiments, typically, the wave propagation speed $c$ may depend on the
wave frequency,   a situation commonly dubbed in General Relativity as  
rainbow spacetimes, a situation which can be indeed also described from a Finslerian perspective \cite{Rainbow}.
Our present  approach and, in particular, our Huygens' envelope principle for wavefronts,  should be considered 
as  the first step towards the description of these more realistic configurations. The literature on the experiments \cite{A1,A2,A3,A4,A5,A6,A7,A8} discusses in details all these points.

 \subsection{Wavefronts in flumes}
 The first, and still more common, type of hydrodynamic analogue gravity model is 
 the case of surface waves in a long and shallow channel flow, a situation having effectively only
 one spatial dimension. Typically, the flow is stationary but its velocity $V$
 depends on the position due to the presence of certain obstacles in the channel
 bottom, see  \cite{A1,A2,A3,A4} for some concrete realizations of this kind
 of experiment. The surface waves propagation velocity $c$ also depends on the position
 along the channel. Horizons for the surface waves can be produced by selecting obstacles
 such that $c^2< V^2$ on some regions along the channel. 
 
 We will consider here the simplest case consisting of $(\mathbb{R},h)$, {\em i.e.} the real line 
  with the standard metric $h$,
 and the Zermelo  
vector field $W(x)$, where $x$ is  coordinate along the channel, with $W^2 < 1$.  
  Let $(\mathbb{R},F)$ be the associated Randers space, where the Randers metric $F$ is
  given by (\ref{Randers1}). Since the Randers space is one-dimensional in this case,
  the wavefronts    will correspond to a set of two points, and we do not need to worry
  about wave rays and their orthogonality to the wavefronts. For the sake of simplicity,
  suppose the waves are emitted at $t=0$ from a single point $q$.  The wavefront at $t=r$ will be given by
$\rho^{-1}(r)=\{p\in \mathbb{R} \ :\ d_F(q,p)=r\}$. Assuming that  
 $\gamma:[0,r]\to \mathbb{R}$ is the unit speed geodesic that realizes this distance,
we have  
\begin{equation}
\label{geod}
1=F(\gamma,\dot{\gamma})=\frac{|\dot{x}| -  W\dot x }{1-W^2 }  ,
\end{equation}
where (\ref{Randers1}) was used. From equation (\ref{geod}), we have that the right ($x_+$) and
 left-moving ($x_-$)  
 wavefronts are governed by the equations 
\begin{equation}
\label{wavefront}
\dot x_\pm = \pm 1 + W(x) ,
\end{equation}
and the wavefront at $t=r$ will be simply $\rho^{-1}(r)=\{x_-(r),x_+(r) \}$. Notice that both equations
 (\ref{wavefront}) are separable and could be solved straightforwardly  by quadrature, but for our purpose here a dynamical analysis for general $W$ typically
   suffices. Since $|W|<1$, there are no fixed points in  (\ref{wavefront}), meaning that
 $x_+$  and
    $x_-$ move continuously towards right and left, respectively.  Let us consider the explicit example of
 the Zermelo vector
 \begin{equation}
 \label{Zerm}
    W(x) = \frac{a}{1+x^2} ,
 \end{equation}
 with $0\le a<1$.
 Its aspect is quite simple (see Fig. \ref{Fig1}), 
\begin{figure}[t]
\hspace{3.5cm} \includegraphics[scale=0.75]{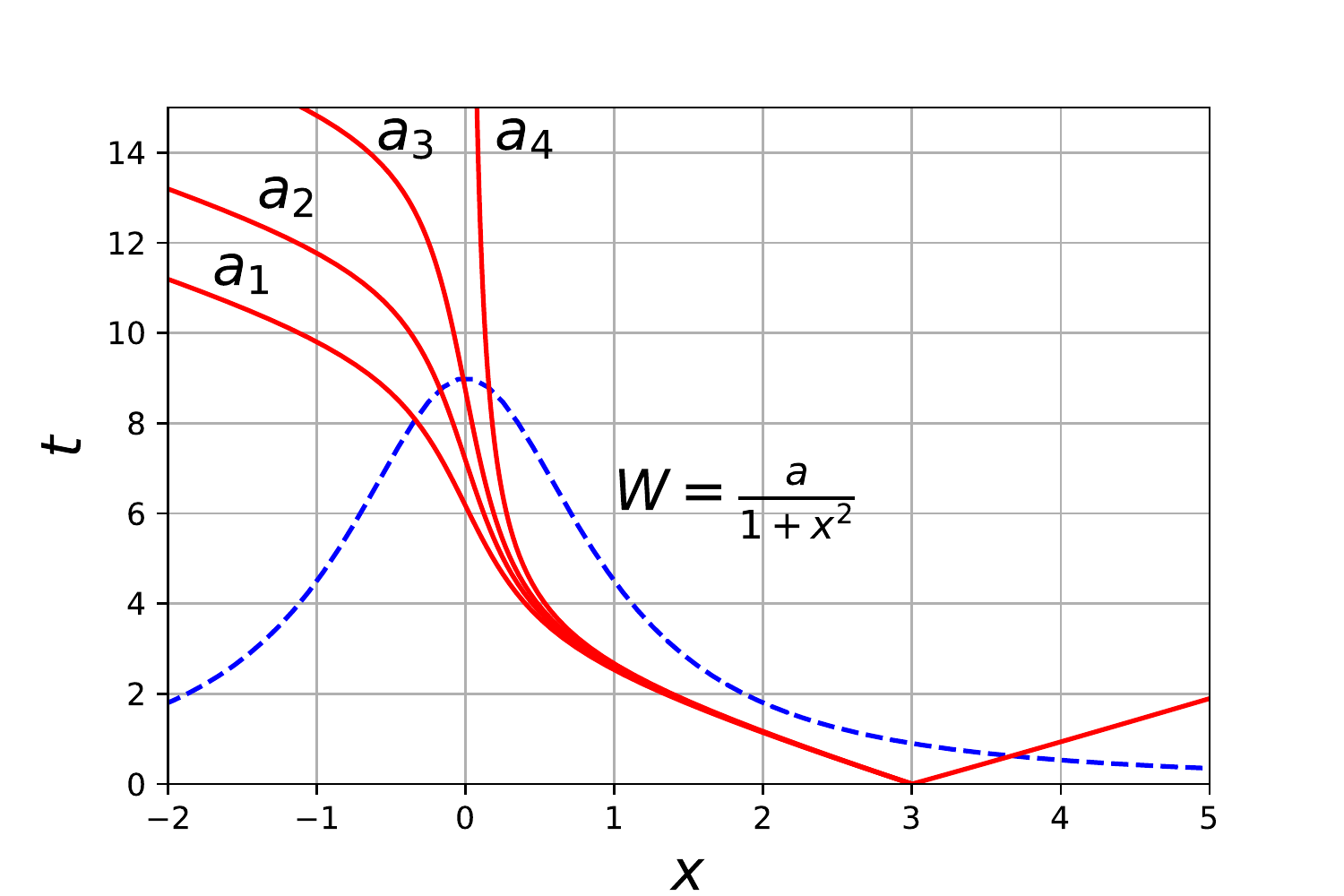}  
\caption{Geodesics (\ref{x+}) and (\ref{x-}) for some specific values of $a$. The blue (traced) curve is the aspect of the Zermelo vector (\ref{Zerm}), without any scale.
The red (continuous) lines corresponds to two sets of right and left-moving geodesics which start  at $t=0$ in  $x_0=-3$ and $x_0=3$. The right-moving geodesics $x_+$ are rather insensitive to the value of $a$. On the other hand, the  left-moving ones that cross the maximum of $W$ at $x=0$ reveal great 
sensitivity, mainly for the cases with $a$ close to 1. The depicted curves correspond to the following values of $a$: $a_1=0.85$, $a_2=0.9$, $a_3=0.94$, and $a_4=0.999$. The existence of a Killing vector in $x=0$ for $a\to 1$ is denounced by the behavior of the left-moving geodesics which start  in $x>0$, see the main text. 
}
\label{Fig1}
\end{figure}
it corresponds to a flume moving to the right-handed direction, with a smooth and non-homogeneous
 velocity attaining its maximum $W=a$ at the origin, which might be caused, for instance, due the presence
 of a smooth obstacle in the channel. We could also add a positive constant to $W$ which would correspond to the flume velocity far from the origin, but for our purposes here this constant is irrelevant and  we set it to zero, without any loss of generality.   For this choice of $W$, equations (\ref{wavefront})
 can be exactly solved as
 \begin{equation}
 \label{x+}
t=x_+ - x_0 - \frac{a}{\sqrt{1+a}}\left( \arctan \frac{x_+}{\sqrt{1+a}}- \arctan \frac{x_0}{\sqrt{1+a}}\right)
 \end{equation}
and
 \begin{equation}
 \label{x-}
 t=      x_0 -x_-  + \frac{a}{\sqrt{1-a}}\left(  \arctan \frac{x_0}{\sqrt{1-a}} - \arctan \frac{x_-}{\sqrt{1-a}}\right)  ,
 \end{equation}
 where we assume, for sake of simplicity and also without loss of generality, that the wavefront started at $t=0$ in $x=x_0$. From (\ref{x+}) and (\ref{x-}), we can draw a $(x,t)$ diagram for the geodesics, see Fig. \ref{Fig1}.
The behavior of the right ($x_+$) and left-moving ($x_-$) geodesics depends  strictly on the value of $x_0$. For $x_0>0$, the right moving geodesics depart  from the maximum of $W$ located at $x=0$ and are rather insensitive on the value of the constant $a$. Exactly the same occurs for the left-moving geodesics starting at $x_0<0$. The situation for the geodesics   crossing $x=0$ is completely different. The left-moving ones starting at $x_0>0$ cross $x=0$ ``against'' the Zermelo vector $W$ and depict a strong sensibility on $a$. In particular, for $a$ very close to 1, they tend  to stay close to $x=0$ for large time intervals. On the other hand, the right-moving geodesics  which started in $x_0<0$, and move in the same direction of $W$, exhibit low sensitive on $a$, they cross $x=0$ without sensible deviations. 
We can understand such differences directly from the Randers metric for our case
\begin{equation}
F(x,y) = \frac{ (1+x^2)}{(1+x^2)^2-a^2}\left(\left(1+x^2 \right)|y|  -ay\right).
\end{equation}
Notice that for $y>0$, the metric near the origin $x=0$ reads
\begin{equation}
\label{h1}
F(x,y) = \frac{y}{1+a} + O(x)
\end{equation}
whereas
\begin{equation}
\label{h2}
F(x,-y) = -\frac{y}{1-a} + O(x).
\end{equation}
It is clear that we will have for $a\to 1$ a manifestation of a Killing horizon in the Randers space, a hypersurface acting as an one-direction membrane, {\em i.e.} an hypersurface which can be crossed only in one direction. We will return to this important point in the last section.   
 
 \subsection{Wavefronts in vortices}  
 The one-dimensional flows of the first example are not sufficient to appreciate 
 all the subtleties   of the Finslerian analysis of the wavefronts. Flows involving
 vortices are very good candidates for our study, since besides of being intrinsically
 higher dimensional, they are indeed important from the experimental point of view
 in analogue gravity, see \cite{A5,A6,A7,A8} for some recent results. We will consider here the simplest possible vortex
 configuration:
 a fluid in 
 a long cylindrical tank  $M$    of radius $ R>0$. We will assume cylindrical symmetry, so  the vertical direction can be neglected and we are left with an effective  two-dimensional spatial problem.
  The pertinent manifold
 for our flow will be   
  \begin{equation}
    M=\left\{(x^1,x^2 )\in \RR^{2}\ :\  {(x^1)}^2+{(x^2)}^2 \le R^2\right\}. 
  \end{equation} 
  It is important to stress that our manifold in this case has a boundary $\partial M$ and that some boundary
  conditions will be needed for wavefronts and geodesics reaching $\partial M$ .   
 The associated 
  Randers space   will be $(M,h)$, where 
  $ h $ is the usual Euclidean two-dimensional metric, with the Zermelo vector field $W = (w^1,w^2)$
corresponding to a rotation flow around the origin. If one wants to keep the cylindrical symmetry, the
more general  Zermelo vector field in this case will be of the type
$W = w(r)XA$, where $X= (x^1,x^2)$, $r^2 = XX^t$, $w(r)$ is a smooth function,  and
$A$ is the two-dimensional rotation generator matrix 
\begin{equation}\label{matrix}
A=\left( {\begin{array}{rr}
	0 & 1 \\
	-1 & 0 \\
	\end{array} } \right).
\end{equation}
The case of constant angular fluid velocity (rigid rotation) corresponds to $w=a$ constant, whereas
the constant tangential  velocity is $w=ar^{-1}$. Notice that 
the dynamical flow associated with such a vector field is  given by   $\varphi_W(t,X)=X{\rm Rot}_r(t)$, where 
\begin{equation}\label{rot}
{\rm Rot}_r(t)=
\left( {\begin{array}{rr}
	\cos (tw(r)) & \sin (tw(r)) \\
	-\sin (tw(r)) & \cos (tw(r)) \\
	\end{array} } \right).
\end{equation}
The Randers metric (\ref{Randers1}) in this case is given by
\begin{equation}
\label{Rand2}
  F(X,Y) = \frac{\left|WY^t \right|}{1 - WW^t}\sqrt{1+\left(1 - WW^t\right)\frac{YY^t}{(WY^t)^2}}
    - \frac{WY^t  }{1 - WW^t}   ,  
\end{equation}
where $Y = (y_1,y_2)\in TM$ is an arbitrary vector and
\begin{equation}
 \lambda = 1 - WW^t = 1 -r^2w^2 ,
\end{equation}
from where we have the restriction $\max | rw    |< 1$. Of course, we have also assumed
$ WY^t\ne 0$. 
 Notice that (\ref{Rand2}) 
clearly resembles the behavior of the Randers metric (\ref{h1}) and (\ref{h2}) of the previous
example. If we have $rw   $ close to 1 for some $r=r_0$, we will have in the neighborhood of this hypersurface
\begin{equation}
 F(X,Y) = \frac{ \left|WY^t \right| - WY^t  }{1 - WW^t} + \frac{YY^t}{2| WY^t | } + O(r-r_0),
\end{equation}
and that it is clear that
for $WY^t > 0$ (corresponding to the vector $Y$ pointing in the same direction of the Zermelo ``wind'' $W$),
the metric is insensitive to the term $(1 -r^2w^2)^{-1}$, in sharp contrast with the situations where
$WY^t< 0$ (the vector $Y$ ``against''   $W$). The hypersurface $r=r_0$ in this case is not
exactly an horizon, since it could indeed be crossed in both direction by, for instance, having
$WY^t > 0$ but with ingoing and outgoing radial directions for $Y$. This kind of
 hypersurface mimics the main
properties of a black hole ergosphere, since it practically favors co-rotating directions for $Y$, as the
counter-rotating ones are strongly affected by the singularity arising from $(1 -r^2w^2)^{-1}$ 
for $rw\to 1$. 
An explicit example for $w(r)$ will help to illustrate such results. Before that, however,
let us notice that the cylindrical symmetry has an important consequence for the wavefronts. 
Let us consider the function 
 $f:M\to\RR$ with $f(x)= r^2  $. Since $df(W) =0 $, we have from Lemma (\ref{nablas}) that
 $F^2(X,\nabla f) = 4f^2$, 
 implying that $f$ is F-transnormal. Hence, by Proposition \ref{propag.water}, we have
 that the circumferences 
 $f^{-1}(t)=\{x\in M\ :\  r^2 = (x^1)^2+(x^2)^2 =t^2\}$
correspond to wavefronts in this Randers space. Of course, due to the cylindrical symmetry, 
such wavefronts originated form a source at the origin $(0,0)$ at $t=0$. The evaluation   of
the wavefronts emitted from an arbitrary points for general $w(r)$ is much more intricate and   involve the 
Finslerian geodesic flow. 

The explicit case we will discuss corresponds to the rigid rotation $w=a$. The main advantage
of this choice is that the Zermelo vector $W$ is a Killing vector of the Euclidean metric, and hence we can use Lemma \ref{geo} to obtain the Finsler geodesics explicitly as 
\begin{equation}
 \gamma_F(t)=\varphi_W(t,\gamma_h(t))=\gamma_h(t){\rm Rot}_a(t) ,
 \end{equation} 
 where $\gamma_h$ are the usual unit speed Euclidean geodesics and ${\rm Rot}_a(t)$ is 
 the matrix (\ref{rot}) for $w=a$. Since the Euclidean geodesics are  
 \begin{equation}
 \gamma_h(t) = (x^1_0,x^2_0) + tV ,
 \end{equation} 
 where $(x^1_0,x^2_0)\in M$ is an arbitrary point and $V$ a unit  vector, one can write
 \begin{equation}
 \label{rays}
 \left(\gamma_F(t) - (x^1_0,x^2_0) {\rm Rot}_a(t) \right) \left(\gamma_F(t) - (x^1_0,x^2_0) {\rm Rot}_a(t) \right)^t = t^2. 
 \end{equation}
 Recalling that the geodesics $\gamma_F$ are the wave rays of the wavefronts, we have from
 (\ref{rays}) that the wavefront emitted at $t=0$ from the point $(x^1_0,x^2_0)\in M$ is an
 expanding circle  with rotating center  $(x^1_0,x^2_0) {\rm Rot}_a(t)  $.
 Moreover, from (\ref{orto}) one can say that the geodesics $\gamma_F$ are orthogonal to each 
of these circles, as one can see by observing that  $F(\gamma')=1$ and that
 $\gamma_F'-W=\gamma_h'(t){\rm Rot}_a(t)$.
Fig. \ref{Fig2} depicts a typical example of wavefronts and geodesics for this system. 
\begin{figure}[tb]
\hspace{4cm}\includegraphics[scale=0.5]{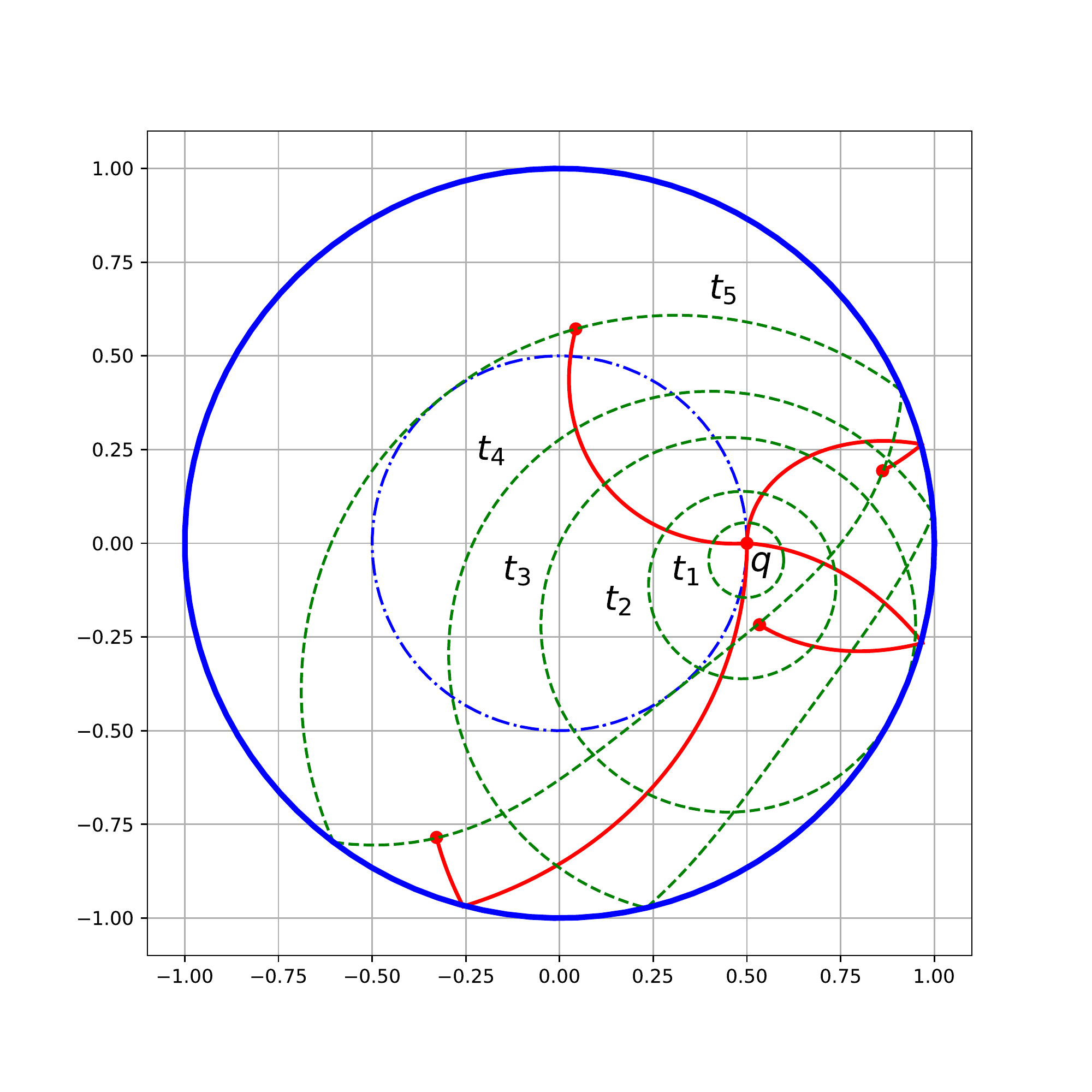} 
\caption{A wave pulse is emitted from the point $q$ at $t=0$ in a Randers space with
a rigid clockwise rotation vector field $W$. The red (solid) curves are some wave rays (Finsler geodesics),
plotted for $0\le t\le t_5$. The green (traced) curves are the wavesfronts at different times $0<t_1<t_2<t_3<t_4<t_5$. 
Perfect reflection on the boundary is assumed. Before the reflection, 
the wavefronts are circles with increasing radius and centers rotating on the blue (dot-traced) circle. After the reflection, they correspond to a circle segment and a caustic (see the text), which centers also rotate along the blue (dot-traced) circle.
   The wave rays are always orthogonal, with respect to the Finlerian structure, to the wavefronts. Due to
   the perfect reflection boundary condition, the wavefronts eventually evolve some singularities, see   the animations
available in the Supplementary Material. 
}
\label{Fig2}
\end{figure}
Since our manifold has a boundary $\partial M$, one needs to specify boundary conditions for geodesics and wavefronts on $\partial M$. We choose to impose perfect reflection on the boundary. The situation for the wavefronts is completely analogous, up to the rigid rotation, to the classical optical  problem of the reflection of spherical waves on a spherical mirror. In particular, after the reflection, our circular wavefronts will form a caustic, see \cite{caust} for a recent approach for the problem. The animations
available in the Supplementary Material depict the typical dynamics of wavefronts and geodesics with perfect reflection
boundary conditions in this Randers space. 
 Before the reflection occurs, 
the wavefronts are circles with increasing radius and which centers rotate 
around the origin $(0,0)$ of $M$ with constant angular velocity $a$. After the reflection, they will correspond to a circular segment  and a caustic, which centers also rotate around the origin. The caustic can be determined by the classical formula
\cite{caust}
\begin{equation}
C(s) = P + \frac{\left(s - |P-X|\right)\left(2(PX^t)P - R^2(P+X) \right)}{|P-X|R^2},
\end{equation}
where $P$ is the point of reflection on the boundary and $X$ the emitting point. The reflection takes place for $s > |P-X|$, {\em i.e.}, for a fixed $s > |P-X|$, $C(s)$ corresponds to the the reflected wavefront (the caustic). Due to the reflection boundary conditions, the wavefronts eventually will contract and give origin to some caustic singularities, see the Supplementary Material and \cite{sing} for further references on these phenomena. 

We are still left with the ergosphere properties of the hypersurface $ar\to 1$. They are illustrated in 
Fig. \ref{Fig3}.
\begin{figure}[t]
\hspace{4cm}\includegraphics[scale=0.5]{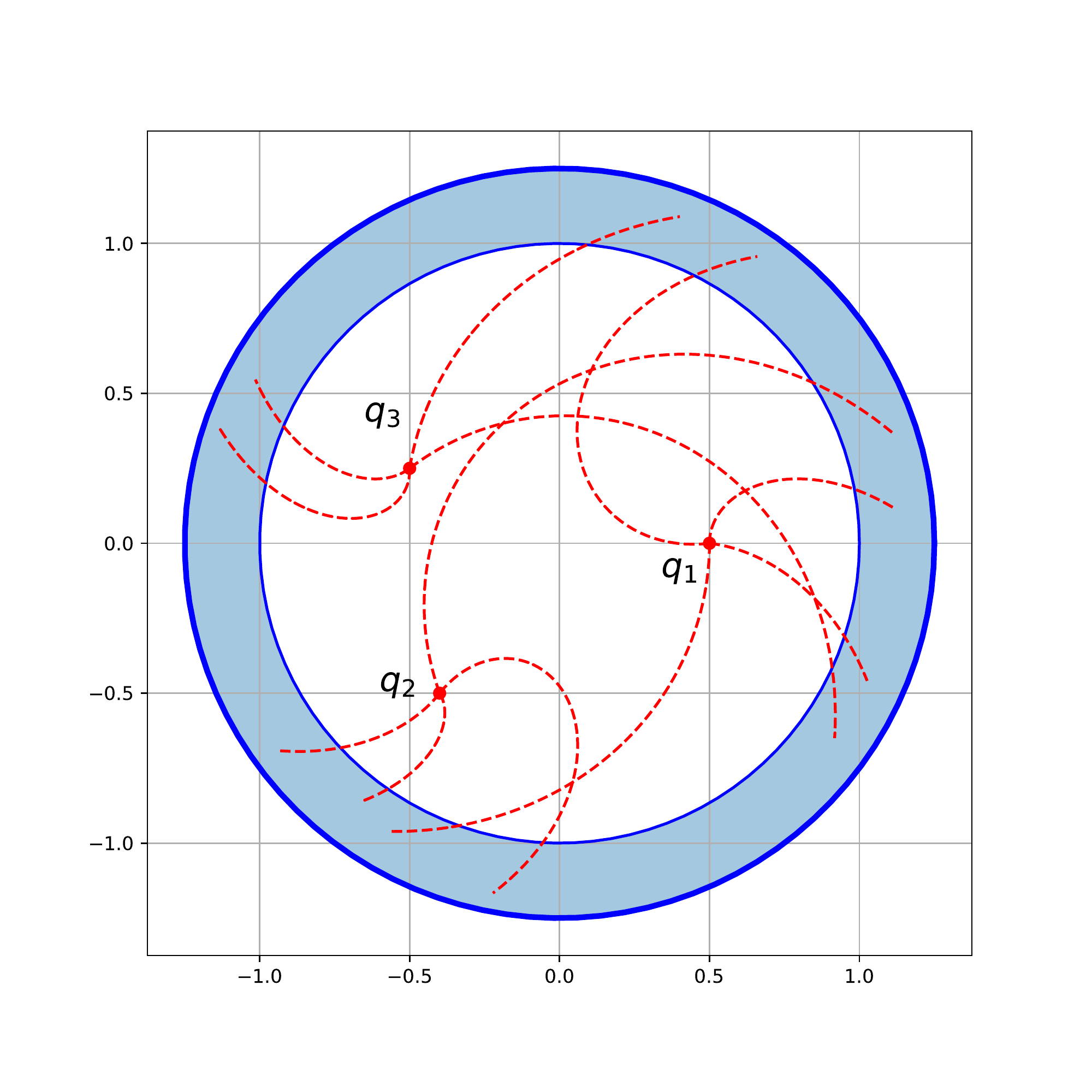} 
\caption{For the case of a  rigid clockwise rotating vector field $W$,
several Finsler geodesics, starting in three different points $q_1, q_2$, and $q_3$,  are depicted.
  The inner circle, which corresponds to $ar\to 1$, mimics an ergosphere. In the outermost
region (shadowed), which rigorously does not belong to our Randers space, no counter-rotating wave rays would be allowed.  No geodesic is allowed to reach the hypersurface $ar\to 1$ in the counter-rotating direction. 
}
\label{Fig3}
\end{figure}
Several geodesics starting in  different points are depicted. All geodesics reach the hypersurface $ar\to 1$ in the co-rotating direction (the Zermelo vector is a rigid clockwise rotation). In the outermost region,    which rigorously does not belong to our Randers space,
no counter-rotating wave rays would be allowed. This is qualitative equivalent to ergospheres in Kerr black holes, where no static observers are allowed since they are inexorably dragged and  co-rotate with the black-hole. 
The existence of ergospheres is intimately connected with superradiant scattering, a phenomenon already
described and detected in analogue models involving surface waves in vortex flows, see \cite{A6}, for
instance.

\section{Final remarks}
 
We have extended to the $n$-dimensional case a recent theorem due to  Markvorsen \cite{markvorsen2016finsler} establishing the validity of the Huygens' envelope principle 
in Finsler spaces. We then apply our results to two explicit cases motivated by recent
results in analogue gravity: the propagation of surface waves in flumes and vortex flows. The Finslerian description associated with the Fermat's principle of least time for the
wave propagation, in both cases, gives rise  to an underlying Randers geometry  and provides a  useful framework for the study of wave rays and wavefronts propagation. 
Interestingly, the spatial regions where $h(W,W)\to 1$ exhibits clearly the distinctive directional properties of
some spacetime causal  structures, namely a Killing horizon for the uni-dimensional flume and an ergosphere for the two-dimensional vortex. However, from the Randers space point of view, we are
confined by construction into the regions where the so-called mild Zermelo wind condition
$h(W,W) < 1$ holds, and hence the full description of these
issues would require abandoning the mild wind condition and the introduction of a Kropina-type metric for the region where $h(W,W) > 1$, see \cite{Kropina} for some recent mathematical
developments in this problem. In such a unified description, we could describe properly both sides
of the spatial hypersurface corresponding to $h(W,W) = 1$. This unified description of Randers
and Kropina spaces is still a quite recent program in Mathematics \cite{Kropina}. 

The analogue gravity examples provide a rather direct application for the Finslerian approach since 
the use of the Fermat's principle        manifestly originates an underlying Randers  geometry. However, the same results would also hold for General Relativity. Consider, for instance, the
Schwarzschild metric in the Gullstrand-Painlev\'e stationary coordinates
\begin{equation}
\label{GP}
ds^2 = -\left(1-\frac{2M}{r} \right)dt^2 + 2\sqrt{\frac{2M}{r}}dtdr + dr^2 + r^2 d\Omega^2,
\end{equation}
where $d\Omega^2$ stands for the usual metric on the unit sphere. Such a metric is, indeed, the
starting point for several hydrodynamic analogies, see \cite{river}, for instance. Ignoring
the angular variables, the null geodesics
of (\ref{GP}) are such that
\begin{equation}
dt = F(r,dr)  = \frac{|dr| +\sqrt{\frac{2M}{r}}dr}{1-\frac{2M}{r}},
\end{equation}
and this is precisely  a Randers metric of the type (\ref{geod}) with Zermelo vector $W= - \sqrt{\frac{2M}{r}}$. Exactly as in the flume case, we have two qualitative different behaviors for ingoing ($dr = dr_-<0$) and outgoing  ($dr=dr_+>0$) null rays, namely
\begin{equation}
\label{diverg}
dt = \frac{dr_\pm}{1 \mp \sqrt{\frac{2M}{r}}}.
\end{equation}
The directional properties of such Randers metric indicate  the presence of a horizon at $r=2M$, since ingoing null rays can cross it smoothly, while outgoing rays experiment  a metric divergence. It is hardly a surprise that Finsler
geometry turns out to be relevant for these directional properties of a spacetime causal structure.
In fact, some recent mathematical results \cite{M1,M2} show that most of causality results are also
valid in a Finslerian framework, under rather weak regularity hypotheses. However, 
the application of Finsler geometry in physical studies of causal structures is still a rather incipient program. Dropping the mild wind condition, which in this case
should allow for a  unified description for the exterior and interior region of the black hole (\ref{GP}), and the
study of the Finslerian curvatures associated to the divergence in (\ref{diverg}), should be the first 
steps
towards a physical Finslerian description of spacetime causal structures. These topics are now under investigation. 

\section*{Supplementary material}

The animations available as  Supplementary material
  at \cite{animation} 
 show the continuous time evolution of the wavefronts and
geodesics of Fig. \ref{Fig2}. One can appreciate the eventual formation of singularities in the caustic associated with the reflection of the wavefronts in the boundary $\partial M$.

\section*{Acknowledgment}
The authors 
acknowledge the financial support of CNPq, CAPES,    and FAPESP (Grant 2013/09357-9).
They also 
wish to thank M.M. Alexandrino, B.O. Alves,   M.A. Javaloyes, and E. Minguzzi for
enlightening discussions.

\end{document}